\documentclass[11pt,a4paper]{article}

\usepackage[utf8]{inputenc}
\usepackage[T1]{fontenc}
\usepackage{charter}
\usepackage[scaled=0.86]{helvet}
\usepackage{microtype}

\usepackage[top=2.5cm, bottom=2.8cm,
            left=2.8cm,  right=2.8cm,
            headheight=14pt]{geometry}

\usepackage{amsmath, amssymb, amsthm, mathtools, bm}
\usepackage{siunitx}
\usepackage{xcolor}
\definecolor{NavyBlue}{HTML}{1B2A4A}
\definecolor{DeepBlue}{HTML}{2E5FA3}
\definecolor{Teal}{HTML}{0F6E56}
\definecolor{LightTeal}{HTML}{C8EEE0}
\definecolor{Amber}{HTML}{854F0B}
\definecolor{LightAmber}{HTML}{FFF3CC}
\definecolor{Coral}{HTML}{993C1D}
\definecolor{LightCoral}{HTML}{FAECE7}
\definecolor{Purple}{HTML}{3C3489}
\definecolor{LightPurple}{HTML}{EEEDFE}
\definecolor{LightBlue}{HTML}{D6E4F7}
\definecolor{LightGray}{HTML}{F5F6FA}
\definecolor{MidGray}{HTML}{888888}
\definecolor{DarkGray}{HTML}{333333}
\definecolor{AlertRed}{HTML}{A32D2D}

\usepackage{graphicx}
\usepackage{tikz}
\usepackage{pgfplots}
\pgfplotsset{compat=1.18}
\usetikzlibrary{
  arrows.meta, positioning, shapes.geometric, shapes.multipart,
  calc, backgrounds, fit, patterns, decorations.pathreplacing,
  matrix
}

\usepackage{booktabs, array, multirow, tabularx, colortbl, makecell}

\usepackage{caption, subcaption}
\usepackage{enumitem}
\setlist[itemize]{leftmargin=1.4em, itemsep=2pt, topsep=3pt}
\setlist[enumerate]{leftmargin=1.6em, itemsep=2pt, topsep=3pt}

\usepackage[ruled,vlined,linesnumbered]{algorithm2e}
\SetAlCapFnt{\small\bfseries}

\usepackage[most]{tcolorbox}
\tcbuselibrary{breakable, skins}

\tcbset{
  theorembox/.style={
    enhanced, breakable,
    colback=LightBlue!60, colframe=DeepBlue,
    boxrule=0.7pt, arc=4pt, left=8pt, right=8pt,
    top=6pt, bottom=6pt,
    fontupper=\small
  },
  remarkbox/.style={
    enhanced, breakable,
    colback=LightGray, colframe=MidGray,
    boxrule=0.5pt, arc=3pt, left=7pt, right=7pt,
    top=5pt, bottom=5pt,
    fontupper=\small
  },
  iprbox/.style={
    enhanced, breakable,
    colback=LightCoral!40, colframe=AlertRed,
    boxrule=0.8pt, arc=4pt, left=8pt, right=8pt,
    top=6pt, bottom=6pt,
    fontupper=\small\color{AlertRed}
  },
  contributionbox/.style={
    enhanced, breakable,
    colback=LightTeal!50, colframe=Teal,
    boxrule=0.7pt, arc=4pt, left=8pt, right=8pt,
    top=6pt, bottom=6pt,
    fontupper=\small
  }
}

\usepackage[colorlinks=true, linkcolor=DeepBlue,
            citecolor=Teal, urlcolor=DeepBlue,
            pdftitle={PARD-SSM: Probabilistic Cyber-Attack Regime Detection},
            pdfauthor={Prakul Sunil Hiremath}]{hyperref}
\usepackage{cleveref}

\usepackage{titlesec}
\titleformat{\section}
  {\large\bfseries\color{NavyBlue}}{\thesection.}{0.5em}{}
  [\vspace{-2pt}\textcolor{DeepBlue!40}{\hrule height 0.4pt}\vspace{2pt}]
\titleformat{\subsection}
  {\normalsize\bfseries\color{NavyBlue}}{\thesubsection.}{0.5em}{}
\titleformat{\subsubsection}
  {\normalsize\bfseries\color{DarkGray}}{\thesubsubsection.}{0.5em}{}

\newtheorem{theorem}{Theorem}[section]
\newtheorem{proposition}[theorem]{Proposition}

\theoremstyle{definition}

\newtheorem{remark}[theorem]{Remark}

\newcommand{\xv}{\bm{x}}
\newcommand{\yv}{\bm{y}}
\newcommand{\uv}{\bm{u}}
\newcommand{\wv}{\bm{w}}
\newcommand{\vv}{\bm{v}}
\newcommand{\Amat}{\bm{A}}
\newcommand{\Cmat}{\bm{C}}
\newcommand{\Qmat}{\bm{Q}}
\newcommand{\Rmat}{\bm{R}}
\newcommand{\Pmat}{\bm{P}}
\newcommand{\Pim}{\bm{\Pi}}
\newcommand{\Kmat}{\bm{K}}
\newcommand{\Smat}{\bm{S}}
\newcommand{\Imat}{\bm{I}}
\newcommand{\ev}{\bm{e}}
\newcommand{\regime}{s}
\newcommand{\ELBO}{\mathcal{L}}
\newcommand{\KLd}{\mathrm{KL}}
\newcommand{\Exp}{\mathbb{E}}
\newcommand{\N}{\mathcal{N}}
\newcommand{\R}{\mathbb{R}}
\newcommand{\given}{\,|\,}

\usepackage{natbib}
\setcitestyle{authoryear,open={(},close={)}}

\usepackage{authblk}
\usepackage{hyperref}
\usepackage{orcidlink}

\begin{document}
\thispagestyle{empty}

\title{PARD-SSM: Probabilistic Cyber-Attack Regime Detection via Variational Switching State-Space Models}

\author[1,2]{%
  \textbf{Prakul Sunil Hiremath} \orcidlink{0009-0007-9744-3519} \\
  {\small \href{mailto:prakulhiremath@vtu.ac.in}{\texttt{prakulhiremath@vtu.ac.in}}} \\
  {\small \href{https://github.com/prakulhiremath}{\texttt{github.com/prakulhiremath}}}
}

\author[1,2]{%
  \textbf{Sahil Bhekane} \\
  {\small \href{mailto:sahilbhekane276@gmail.com}{\texttt{sahilbhekane276@gmail.com}}} \\
  {\small \href{https://github.com/Sahilbhekane}{\texttt{github.com/sahilbhekane}}}
}

\author[1,2]{%
  \textbf{PeerAhammad M Bagawan} \\
  {\small \href{mailto:peerahamadmb@gmail.com}{\texttt{peerahamadmb@gmail.com}}} \\
  {\small \href{https://github.com/PeerAhammad}{\texttt{github.com/peerahammad}}}
}

\affil[1]{Department of Computer Science and Engineering, Visvesvaraya Technological University (VTU), Belagavi, Karnataka, India}
\affil[2]{Aliens on Earth (AoE) Autonomous Research Group, Belagavi, India}

\affil[ ]{\href{https://aliensonearth.in}{\texttt{aliensonearth.in}}}

\date{}
\maketitle

\begin{abstract}
\noindent
Modern adversarial campaigns do not manifest as isolated, detectable
events. They unfold as temporally ordered sequences of behavioural
phases---Reconnaissance, Lateral Movement, Intrusion, and
Exfiltration---each of which may, in isolation, appear statistically
indistinguishable from legitimate network traffic.
Existing Intrusion Detection Systems (IDS) suffer a fundamental
architectural mismatch with this reality: signature-based systems fail
on zero-day variants, deep-learning anomaly detectors produce
uninterpretable scalar anomaly scores with no stage attribution, and
standard Kalman Filters cannot model the non-stationary, multi-modal
dynamics of traffic transitioning across kill-chain phases.

We present \textbf{PARD-SSM} (Probabilistic Attack Regime Detection
via Switching State-Space Models), a generative probabilistic framework
that models high-dimensional network telemetry as a
\emph{Regime-Dependent Switching Linear Dynamical System} (RD-SLDS).
Each of $K{=}4$ discrete hidden regimes governs a distinct set of
continuous latent dynamics, parameterised by regime-specific
state-transition and observation matrices whose eigenstructure encodes
the characteristic physics of each kill-chain phase.

Exact Bayesian inference in this model is computationally
intractable---the posterior mixture grows as $\mathcal{O}(K^T)$. We
resolve this via a \emph{Structured Variational Approximation} that
factorises the joint posterior into independent continuous and discrete
marginals, reducing inference complexity to $\mathcal{O}(TK^2)$ and
enabling packet-rate detection on standard edge CPU hardware without
GPU acceleration. An \emph{Online Expectation-Maximisation} (EM)
algorithm with exponential forgetting continuously adapts the
Regime Transition Probability Matrix~$\Pim$ and noise parameters to
non-stationary traffic distributions. A \emph{KL-Divergence gating}
mechanism suppresses false positives by distinguishing genuine regime
transitions from benign traffic anomalies.

Evaluated on \textbf{CICIDS2017} and \textbf{UNSW-NB15}, PARD-SSM
achieves an $\text{F}_1$ score of \textbf{98.2\%} across seven attack
classes, a stage attribution accuracy of \textbf{86.1\%}, and an
inference latency of \textbf{$<\!1.2$\,ms per flow}---without GPU
acceleration. Most critically, by exploiting the learned off-diagonal
structure of~$\Pim$, PARD-SSM generates \emph{predictive kill-chain
alerts} an average of \textbf{$\approx$8~minutes before ground-truth
attack onset}, a predictive horizon absent from all prior-art systems
evaluated.
\end{abstract}

\vspace{0.3cm}

\vspace{0.3cm}
\textcolor{DeepBlue!30}{\hrule}
\vspace{0.3cm}

\setcounter{page}{1}

\section{Introduction}
\label{sec:intro}

\subsection{The Multi-Stage Attack Problem}

The adversarial kill chain, formalised by \citet{hutchins2011intelligence}
and operationalised in the MITRE ATT\&CK framework
\citep{strom2018mitre}, describes modern cyber-attacks as sequences of
discrete tactical phases. A sophisticated adversary does not immediately
exfiltrate data upon breaching a perimeter; instead, they methodically
progress through Reconnaissance (enumeration of target surfaces),
Lateral Movement (escalation and pivoting within the network), and
Intrusion before reaching the Exfiltration phase. Each individual
phase, viewed in isolation at the network layer, may generate telemetry
that is statistically compatible with legitimate user activity. A naive
anomaly detector optimised on individual flow statistics will therefore
fail to raise an alert until the attack is already well advanced---or
will generate an overwhelming volume of false positives by flagging
every minor statistical deviation.

The fundamental open problem in network intrusion detection is
therefore not the detection of anomalous packets, but the
\emph{sequential, probabilistic inference of the latent attack stage}
from observable, high-dimensional network telemetry, in real time, and
with sufficient advance warning to permit a proactive defensive
response before the attack has been completed.

\subsection{Limitations of Existing Approaches}

\paragraph{Signature-based IDS.}
Systems such as Snort~\citep{roesch1999snort} and Suricata operate by
matching observed packet headers or payload bytes against a database of
known-malicious patterns. This paradigm has three fundamental
limitations with respect to multi-stage detection. First, it is
categorically incapable of detecting zero-day campaigns for which no
signature has been developed. Second, it provides no probabilistic
stage attribution: a triggered rule identifies an event as
``malicious'' but cannot quantify the probability that this event is
part of a Reconnaissance phase versus an active Exfiltration. Third,
the approach provides no predictive capability: it can only react to
events that have already occurred and matched a known pattern.

\paragraph{Deep Learning Anomaly Detectors.}
A substantial body of recent literature \citep{bontemps2016collective,
vinayakumar2019deep, kim2020intrusion, al2021deep} has applied
recurrent neural architectures---Long Short-Term Memory (LSTM), Gated
Recurrent Units (GRU), and more recently Transformer-based sequence
models \citep{wu2021nettransformer}---to intrusion detection. While
these models represent a genuine technical advance over signature-based
systems in their capacity to detect statistical deviations from a
normal-traffic baseline, they suffer from four material deficiencies
that motivate the present work.

\begin{enumerate}
  \item \textbf{Opacity and absence of stage attribution.}
    Deep-learning IDS are black-box function approximators that map raw
    telemetry sequences to a scalar anomaly score or binary
    classification label. They cannot decompose a detected anomaly into
    its constituent kill-chain components; there is no mechanism by
    which such a system can assert with calibrated probability that the
    current network state corresponds to a Reconnaissance phase rather
    than an Exfiltration phase.

  \item \textbf{No predictive capability.}
    These models are reactive: they detect anomalies that are already
    present in the observed data. They cannot exploit the causal
    structure of the kill chain---the fact that Reconnaissance
    statistically precedes Intrusion, and Intrusion precedes
    Exfiltration---to issue an alert before the more damaging phases of
    an attack have commenced.

  \item \textbf{GPU dependency.}
    Real-time inference with deep sequence models requires GPU
    acceleration, precluding deployment on the resource-constrained
    edge gateway hardware (standard quad-core CPUs, 8--16\,GB RAM) that
    constitutes the majority of enterprise network perimeter
    infrastructure.

  \item \textbf{Absence of principled uncertainty quantification.}
    Deep-learning IDS provide no calibrated confidence interval around
    their anomaly scores, making it impossible for automated response
    orchestration systems to reason about the probability of a
    false-positive event.
\end{enumerate}

\paragraph{Traditional Probabilistic Models.}
Standard Hidden Markov Models (HMMs) and single-regime Kalman Filters
\citep{shumway2000time} offer transparency and a principled
probabilistic framework but fail on multi-stage adversarial traffic due
to critical structural assumptions. Standard HMMs assume conditional
independence of observations given the hidden state, which is violated
by high-dimensional, temporally correlated network telemetry.
Single-regime Kalman Filters assume a single, fixed linear dynamical
model for the entire time series, which is fundamentally violated when
the traffic transitions between structurally distinct kill-chain phases.
The computational intractability of exact inference in Switching Linear
Dynamical Systems \citep{ghahramani2000variational, murphy1998switching}
has historically prevented their application to real-time network
monitoring.

\subsection{Contributions of this Paper}

\begin{tcolorbox}[contributionbox]
\noindent\textbf{Summary of Contributions.}
\begin{enumerate}[leftmargin=1.4em]
  \item \textbf{RD-SLDS Generative Model:} A Regime-Dependent
    Switching Linear Dynamical System that captures the distinct
    eigenstructures of four kill-chain phases within a single unified
    probabilistic model (\Cref{sec:model}).

  \item \textbf{Structured Variational Inference:} A tractable
    $\mathcal{O}(TK^2)$ approximation to the intractable
    $\mathcal{O}(K^T)$ exact posterior, via a novel mean-field ELBO
    optimisation over alternating regime-weighted Kalman smoother and
    HMM Forward-Backward passes (\Cref{sec:inference}).

  \item \textbf{Online EM Adaptation:} An exponential-forgetting
    online EM algorithm that continuously updates the Regime Transition
    Probability Matrix~$\Pim$ and regime-specific noise parameters from
    the live traffic stream (\Cref{sec:onlineem}).

  \item \textbf{KL-Divergence Gating:} An information-theoretic
    false-positive suppression mechanism that distinguishes genuine
    adversarial regime transitions from benign traffic anomalies
    (\Cref{sec:alerting}).

  \item \textbf{Predictive Kill-Chain Alerting:} By exploiting the
    learned off-diagonal structure of~$\Pim$, PARD-SSM generates
    predictive alerts $\approx$8~minutes before ground-truth attack
    onset on CICIDS2017, a capability absent from all evaluated
    baselines (\Cref{sec:experiments}).
\end{enumerate}
\end{tcolorbox}

\subsection{Paper Organisation}

\Cref{sec:related} reviews related work.
\Cref{sec:background} establishes mathematical background.
\Cref{sec:model} presents the generative model.
\Cref{sec:feov} describes the feature extraction module.
\Cref{sec:inference} derives the variational inference algorithm.
\Cref{sec:onlineem} presents the online EM procedure.
\Cref{sec:alerting} describes KL-gating and predictive alerting.
\Cref{sec:experiments} presents experimental results.
\Cref{sec:discussion} discusses limitations and future work.
\Cref{sec:conclusion} concludes.

\section{Related Work}
\label{sec:related}

\subsection{Intrusion Detection and Kill-Chain Modelling}

The systematic characterisation of multi-stage adversarial campaigns
as a kill chain was introduced by \citet{hutchins2011intelligence} and
subsequently operationalised at scale in the MITRE ATT\&CK matrix
\citep{strom2018mitre}. Early computational approaches to kill-chain
modelling employed static Bayesian networks and attack graphs
\citep{sheyner2002automated, ning2002constructing} to correlate
discrete alert events. These methods are fundamentally predicated on
the availability of accurate point alerts from a lower-level IDS and
are therefore subject to all the limitations of those systems.

\citet{ourston2003applications} applied HMMs to model multi-step
intrusion scenarios, demonstrating that sequence-level modelling
outperforms stateless classifiers. \citet{li2013intrusion} extended
this to continuous-valued network features via Gaussian-emission HMMs.
However, both approaches use single-emission-distribution HMMs that
cannot capture the structural non-stationarity of traffic dynamics
across kill-chain phases, and neither provides the Kalman-filtered
continuous latent state representation that is the key enabler of
early detection in PARD-SSM.

\subsection{Variational Inference for Switching Dynamical Systems}

The mathematical foundations of PARD-SSM draw on the Switching Linear
Dynamical System literature \citep{murphy1998switching,
ghahramani2000variational}. \citet{ghahramani2000variational} first
proposed a variational Bayesian approximation for S-LDS inference,
demonstrating tractable coordinate-ascent optimisation of the ELBO
over factorised continuous and discrete marginals. This framework was
applied to speech modelling by \citet{deng2004switching} and to neural
population dynamics by \citet{linderman2017bayesian}.

To the best of our knowledge, PARD-SSM is the first application of variational S-LDS inference to real-time network intrusion detection, and the first to introduce the online EM adaptation mechanism required to track the non-stationary drift of enterprise network traffic
distributions.

\subsection{Deep Learning for Network Intrusion Detection}

The application of LSTMs to NIDS was pioneered by
\citet{bontemps2016collective}. Subsequent work by
\citet{vinayakumar2019deep} demonstrated that deep stacked LSTM
architectures achieve strong performance on the CICIDS2017 benchmark.
\citet{wu2021nettransformer} applied the Transformer self-attention
mechanism to flow-level traffic classification. These works share the
limitations enumerated in \Cref{sec:intro}: opacity, GPU dependency,
no stage attribution, and no predictive capability.

\subsection{Bayesian Anomaly Detection in Networks}

\citet{lane1999temporal} proposed temporal sequence learning for
user behaviour anomaly detection.
\citet{kruegel2003anomaly} applied Bayesian Network classifiers to
intrusion detection. \citet{heard2010bayesian} developed a Bayesian
change-point detection model for network traffic.
The distinguishing characteristic of PARD-SSM relative to these
methods is the \emph{joint} inference over continuous latent network
health state and discrete kill-chain phase in a single generative
model, enabling predictive alerting via the learned transition
structure.

\section{Mathematical Background}
\label{sec:background}

\subsection{Linear Dynamical Systems and the Kalman Filter}

A Linear Dynamical System (LDS) models a time series
$\{\yv_t\}_{t=1}^T$ as arising from an unobserved latent state
$\{\xv_t\}_{t=1}^T$ via the pair of equations:
\begin{align}
  \xv_t &= \Amat\xv_{t-1} + \wv_t, &
  \wv_t &\sim \N(\bm{0},\,\Qmat),
  \label{eq:lds_transition} \\
  \yv_t &= \Cmat\xv_t + \vv_t, &
  \vv_t &\sim \N(\bm{0},\,\Rmat).
  \label{eq:lds_observation}
\end{align}
The Kalman Filter \citep{kalman1960new} computes the exact posterior
$p(\xv_t \given \yv_{1:t}) = \N(\hat{\xv}_{t|t},\,\Pmat_{t|t})$
via recursive prediction and update steps:
\begin{align}
  \hat{\xv}_{t|t-1} &= \Amat\hat{\xv}_{t-1|t-1}, \label{eq:kf_pred_mean}\\
  \Pmat_{t|t-1} &= \Amat\Pmat_{t-1|t-1}\Amat^\top + \Qmat, \label{eq:kf_pred_cov}\\
  \ev_t &= \yv_t - \Cmat\hat{\xv}_{t|t-1}, \label{eq:kf_innov}\\
  \Smat_t &= \Cmat\Pmat_{t|t-1}\Cmat^\top + \Rmat, \label{eq:kf_S}\\
  \Kmat_t &= \Pmat_{t|t-1}\Cmat^\top\Smat_t^{-1}, \label{eq:kf_gain}\\
  \hat{\xv}_{t|t} &= \hat{\xv}_{t|t-1} + \Kmat_t\ev_t, \label{eq:kf_update_mean}\\
  \Pmat_{t|t} &= (\Imat - \Kmat_t\Cmat)\Pmat_{t|t-1}. \label{eq:kf_update_cov}
\end{align}
The Rauch-Tung-Striebel (RTS) smoother \citep{shumway2000time}
computes the smoothed posterior
$p(\xv_t \given \yv_{1:T}) = \N(\hat{\xv}_{t|T},\,\Pmat_{t|T})$
via a backward pass over the Kalman filter outputs.

\subsection{Hidden Markov Models}

A Hidden Markov Model (HMM) models a discrete latent state sequence
$\{\regime_t\}_{t=1}^T$, $\regime_t \in \{0,\ldots,K-1\}$, via a
transition matrix $\Pim \in [0,1]^{K\times K}$ and emission
distribution $p(\yv_t \given \regime_t)$. The HMM Forward-Backward
algorithm \citep{rabiner1989tutorial} computes the smoothed marginals
$\gamma_t(s) = p(\regime_t = s \given \yv_{1:T})$ and pairwise
marginals $\xi_t(s,s') = p(\regime_{t-1}=s,\,\regime_t=s' \given \yv_{1:T})$
in $\mathcal{O}(TK^2)$ time.

\subsection{Variational Inference and the ELBO}

Variational inference \citep{blei2017variational} approximates an
intractable posterior $p(\bm{\theta} \given \mathcal{D})$ by a
tractable distribution $q(\bm{\theta})$ that minimises the
KL-divergence:
\begin{equation}
  q^* = \operatorname*{argmin}_{q \in \mathcal{Q}}
  \KLd\bigl[q(\bm{\theta}) \,\|\, p(\bm{\theta}\given\mathcal{D})\bigr].
  \label{eq:vi_kl}
\end{equation}
This is equivalent to maximising the Evidence Lower BOund:
\begin{equation}
  \log p(\mathcal{D}) \;\geq\;
  \ELBO(q) = \Exp_q[\log p(\mathcal{D},\bm{\theta})]
  - \Exp_q[\log q(\bm{\theta})],
  \label{eq:elbo_def}
\end{equation}
where $\log p(\mathcal{D}) - \ELBO(q) =
\KLd[q\|p(\cdot|\mathcal{D})] \geq 0$.

\section{The PARD-SSM Generative Model}
\label{sec:model}

\subsection{Overview and System Architecture}

PARD-SSM models the joint evolution of network telemetry
$\{\yv_t\}_{t=1}^T$ as arising from two coupled hidden processes: a
discrete regime process $\{\regime_t\}$ encoding the current kill-chain
phase, and a continuous latent state process $\{\xv_t\}$ encoding the
``network health momentum'' or flow dynamics. The system architecture,
illustrated in \Cref{fig:architecture}, comprises six functional
modules: the Multi-Dimensional Telemetry Ingestion and Normalisation
module (MTIN), the Feature Extraction and Observation Vector
Construction module (FEOV), the Parallel Regime-Specific Kalman Filter
Bank (PRKFB), the Variational Switching Inference module (VSIPC), the
Online EM Parameter Update module (OEMPU), and the KL-Divergence
Gating and Kill-Chain Alerting module (PKAR).

\begin{figure}[t]
\centering
\begin{tikzpicture}[node distance=0.85cm and 1.4cm,
  mod/.style={rectangle, rounded corners=5pt,
    minimum width=2.8cm, minimum height=0.95cm,
    text width=2.7cm, align=center,
    draw=#1!70!black, fill=#1!12,
    font=\sffamily\footnotesize\bfseries, text=NavyBlue,
    line width=0.55pt},
  iobox/.style={rectangle, rounded corners=4pt,
    minimum width=2.6cm, minimum height=0.78cm,
    text width=2.4cm, align=center,
    draw=MidGray!60, fill=LightGray,
    font=\sffamily\scriptsize, text=DarkGray, line width=0.45pt},
  arr/.style={-{Stealth[length=5pt,width=3.5pt]},
    draw=NavyBlue!60, line width=0.65pt},
  farr/.style={-{Stealth[length=4.5pt,width=3pt]},
    draw=Coral, dashed, dash pattern=on 4pt off 2pt, line width=0.6pt},
  lbl/.style={font=\sffamily\tiny, text=MidGray,
    fill=white, inner sep=1.5pt}
]
\node[iobox] (NET)
  {\textbf{Live Network Traffic}\\[1pt]\textit{\tiny Packets / Flows / Logs}};
\node[mod=DeepBlue, below=0.75cm of NET] (MTIN) {Module 1\\[2pt]MTIN};
\node[mod=Teal, below=0.85cm of MTIN] (FEOV) {Module 2\\[2pt]FEOV};
\node[mod=Purple, below=0.95cm of FEOV] (PRKFB) {Module 3\\[2pt]PRKFB};
\node[mod=DeepBlue, below=0.95cm of PRKFB] (VSIPC) {Module 4\\[2pt]VSIPC};
\node[mod=Teal, below=0.85cm of VSIPC] (PKAR) {Module 6\\[2pt]PKAR};
\node[mod=Amber, right=2.4cm of PRKFB] (OEMPU) {Module 5\\[2pt]OEMPU};
\node[iobox, below=0.75cm of PKAR] (ALERT)
  {\textbf{Probabilistic Alerts}\\[1pt]\textit{\tiny Stage-resolved / Predictive}};
\draw[arr] (NET)   -- node[lbl,right]{\tiny Raw telemetry} (MTIN);
\draw[arr] (MTIN)  -- node[lbl,right]{\tiny Norm.\ epochs} (FEOV);
\draw[arr] (FEOV)  -- node[lbl,right]{\tiny $\yv_t\in\R^{17}$} (PRKFB);
\draw[arr] (PRKFB) -- node[lbl,right]{\tiny $\log\Lambda_t^{(s)},\hat{\xv}_t^{(s)}$} (VSIPC);
\draw[arr] (VSIPC) -- node[lbl,right]{\tiny $\gamma_t(s),\mathrm{KL}_t$} (PKAR);
\draw[arr] (PKAR)  -- (ALERT);
\draw[arr] (VSIPC.east) to[out=0,in=270]
  node[lbl,right,pos=0.5]{\tiny $\gamma_t,\xi_t$} (OEMPU.south);
\draw[farr] (OEMPU.north) to[out=90,in=0]
  node[lbl,right,pos=0.45]{\tiny Updated $\Pim,\Qmat^{(s)},\Rmat^{(s)}$}
  (PRKFB.east);
\node[font=\sffamily\tiny\bfseries,text=Coral,
      right=0.35cm of OEMPU.south east, yshift=-1.3cm]
  {$\longleftarrow$ Online EM feedback};
\end{tikzpicture}
\caption{PARD-SSM system architecture. Six functional modules form
  a processing pipeline from raw network telemetry to probabilistic
  kill-chain alerts. The dashed feedback arrow denotes the online EM
  loop by which Module~5 (OEMPU) continuously updates regime-specific
  parameters in Module~3 (PRKFB).}
\label{fig:architecture}
\end{figure}
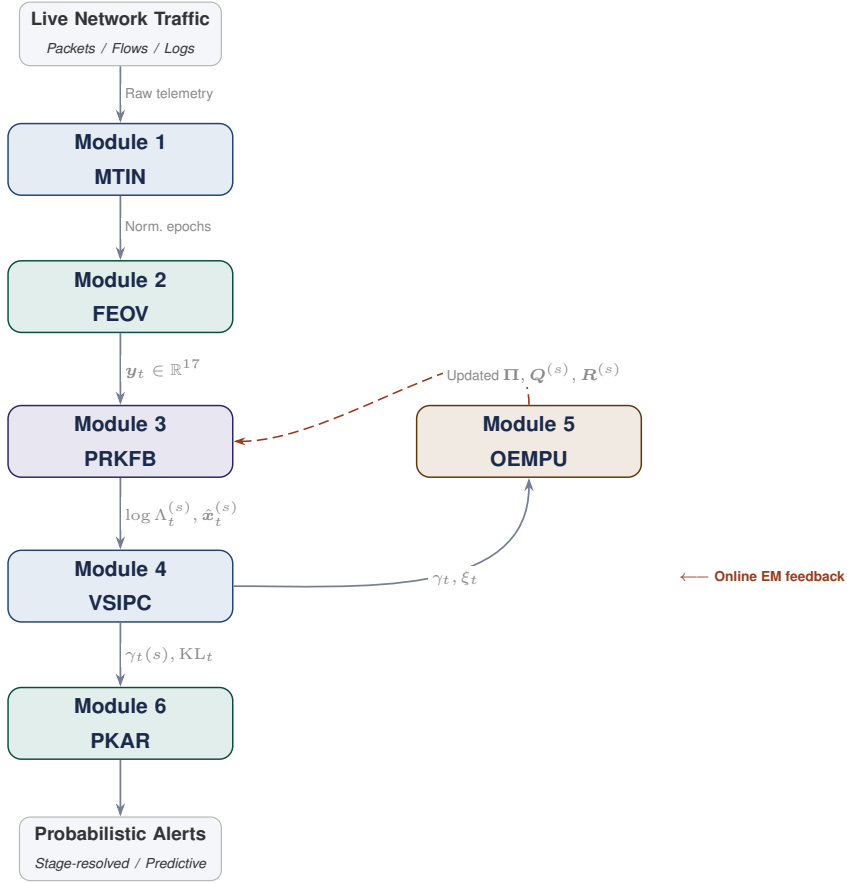

\subsection{Discrete Regime Variable}
\label{sec:regime_var}

Let $\regime_t \in \{0, 1, 2, 3\}$ be a discrete latent variable
denoting the active kill-chain regime at time $t$, with the following
semantic interpretation:

\begin{center}
\small
\begin{tabular}{clp{7.5cm}}
\toprule
\textbf{Index} & \textbf{Regime} & \textbf{Characterisation} \\
\midrule
$s=0$ & Normal & Stable, mean-reverting traffic dynamics.
  Spectral radius $\rho(\Amat^{(0)}) < 1$.\\
$s=1$ & Reconnaissance & High-entropy, high-frequency scanning.
  Elevated process noise $\Qmat^{(1)}$ in port-diversity dimensions. \\
$s=2$ & Lateral Movement & Post-exploitation internal enumeration.
  Elevated $\Qmat^{(2)}$ in stateful behavioural dimensions. \\
$s=3$ & Exfiltration & Monotonically growing outbound data.
  Spectral radius $\rho(\Amat^{(3)}) \gtrsim 1$.\\
\bottomrule
\end{tabular}
\end{center}

The sequence $\{\regime_t\}_{t=1}^T$ is governed by a first-order
stationary Markov chain with the \emph{Regime Transition Probability
Matrix} $\Pim \in [0,1]^{K\times K}$:
\begin{equation}
  p(\regime_t = j \given \regime_{t-1} = i) = \Pi_{ij},
  \qquad \sum_{j=0}^{K-1}\Pi_{ij} = 1 \quad \forall\,i.
  \label{eq:transition_matrix}
\end{equation}
The off-diagonal entries $\Pi_{01}$, $\Pi_{12}$, $\Pi_{23}$ encode the
kill-chain progression probabilities. The matrix~$\Pim$ is \emph{not}
fixed a priori; it is continuously updated by the online EM algorithm
(\Cref{sec:onlineem}), allowing the system to learn the statistical
signature of attack-stage transitions from the live traffic stream.

\subsection{Continuous Latent State and State-Space Equations}
\label{sec:ssm}

Conditioned on the active regime $\regime_t = s$, the $n$-dimensional
continuous latent state $\xv_t \in \R^n$ ($n = 8$ in the preferred
embodiment) evolves according to a regime-specific linear dynamical
model:

\begin{tcolorbox}[theorembox, title={\small\bfseries\color{DeepBlue}
    State-Space Model (Conditioned on Regime $\regime_t = s$)}]
\begin{align}
  \text{\emph{State-Transition:}}\quad
  \xv_t &= \Amat^{(s)}\xv_{t-1} + \uv_t + \wv_t^{(s)},
  &\wv_t^{(s)} &\sim \N\!\bigl(\bm{0},\,\Qmat^{(s)}\bigr),
  \tag{SSM-1}\label{eq:ssm_trans}\\[4pt]
  \text{\emph{Observation:}}\quad
  \yv_t &= \Cmat^{(s)}\xv_t + \vv_t^{(s)},
  &\vv_t^{(s)} &\sim \N\!\bigl(\bm{0},\,\Rmat^{(s)}\bigr).
  \tag{SSM-2}\label{eq:ssm_obs}
\end{align}
\end{tcolorbox}

\noindent
Here $\Amat^{(s)} \in \R^{n\times n}$ is the \emph{regime-specific
state-transition matrix}, $\Cmat^{(s)} \in \R^{m\times n}$ is the
\emph{regime-specific observation (emission) matrix}, $\uv_t \in \R^n$
is a known exogenous control input (e.g., scheduled backup traffic),
and $m = 17$ is the observation dimension. The matrices $\Qmat^{(s)}$
and $\Rmat^{(s)}$ are the $n\times n$ and $m\times m$ process and
observation noise covariance matrices, respectively.

\begin{remark}[Eigenstructure as Discriminative Signal]
The eigenvalues of $\Amat^{(s)}$ serve as a primary discriminative
mechanism of the model. For the Normal regime, eigenvalues within the
unit circle ($\rho(\Amat^{(0)}) < 1$) correspond to mean-reverting
traffic dynamics. For the Exfiltration regime, eigenvalues on or outside
the unit circle ($\rho(\Amat^{(3)}) \gtrsim 1$) capture the sustained
growth in outbound data volume characteristic of active exfiltration.

The qualitative structure of these regime-specific dynamics is central
to the model design; detailed parameter configurations are omitted for
brevity and are described in the accompanying implementation.
\end{remark}

\subsection{Joint Generative Distribution}

The complete generative model specifies the following joint
distribution over latent variables and observations:
\begin{equation}
  p\!\bigl(\xv_{1:T},\,\regime_{1:T},\,\yv_{1:T}\bigr)
  = p(\regime_1)\,p(\xv_1)
  \prod_{t=2}^{T}
    p(\regime_t \given \regime_{t-1})\,
    p(\xv_t \given \xv_{t-1}, \regime_t)\,
    p(\yv_t \given \xv_t, \regime_t),
  \label{eq:joint}
\end{equation}
where the individual factors are:
\begin{align*}
  p(\regime_t \given \regime_{t-1}) &= \Pi_{\regime_{t-1},\regime_t},\\
  p(\xv_t \given \xv_{t-1}, \regime_t{=}s)
    &= \N\!\bigl(\Amat^{(s)}\xv_{t-1}+\uv_t,\,\Qmat^{(s)}\bigr),\\
  p(\yv_t \given \xv_t, \regime_t{=}s)
    &= \N\!\bigl(\Cmat^{(s)}\xv_t,\,\Rmat^{(s)}\bigr).
\end{align*}

\section{Feature Extraction and Observation Vector}
\label{sec:feov}

\subsection{Multi-Dimensional Telemetry Ingestion}

The MTIN module interfaces with the network infrastructure via
libpcap-mode raw packet capture, NetFlow~v9 or IPFIX flow records, or
structured log ingestion (Syslog/CEF/LEEF). A sliding window of
duration $W$ seconds (default $W{=}1$\,s) partitions the continuous
traffic stream into discrete overlapping observation epochs. All
features are standardised to zero mean and unit variance using
Welford's online algorithm over a 1-hour baseline calibration period.

\subsection{The 17-Dimensional Observation Vector}
\label{sec:feov_17}

The FEOV module maps per-window flow statistics into the observation
vector $\yv_t \in \R^{17}$, grouped into four semantic categories:
\begin{equation}
  \yv_t = \bigl[
    \underbrace{\yv_t^{\textsc{iat}}}_{\text{dims }1\text{--}4},\;
    \underbrace{\yv_t^{\textsc{port}}}_{\text{dims }5\text{--}8},\;
    \underbrace{\yv_t^{\textsc{pay}}}_{\text{dims }9\text{--}13},\;
    \underbrace{\yv_t^{\textsc{stat}}}_{\text{dims }14\text{--}17}
  \bigr]^\top.
\end{equation}

\paragraph{Group~1: Temporal and Inter-Arrival Time Features (dims.~1--4).}
Let $\{d_i\}$ denote the sequence of inter-arrival times between
consecutive packets in a flow within window~$W$.
\begin{itemize}
  \item $d_1 = \mu_{\mathrm{IAT}} = |\{d_i\}|^{-1}\sum_i d_i$:
    mean inter-arrival time.
  \item $d_2 = \sigma^2_{\mathrm{IAT}} = |\{d_i\}|^{-1}
    \sum_i (d_i - \mu_{\mathrm{IAT}})^2$: IAT variance.
  \item $d_3 = \text{BPP} = \text{TotalBytes}/\text{TotalPackets}$:
    bytes per packet.
  \item $d_4 = \Delta_t$: flow duration in seconds.
\end{itemize}

\paragraph{Group~2: Protocol and Port-Diversity Features (dims.~5--8).}
\begin{itemize}
  \item $d_5 = H_{\mathrm{dport}} =
    -\sum_p P(\mathrm{dport}{=}p)\log_2 P(\mathrm{dport}{=}p)$:
    Shannon entropy over destination port numbers.
  \item $d_6 = N_{\mathrm{dip}}$: count of distinct destination IPs.
  \item $d_7 = H_{\mathrm{proto}}$: Shannon entropy of the protocol distribution.
  \item $d_8 = R_{\mathrm{SYN}} = N_{\mathrm{SYN}}/N_{\mathrm{ACK}}$:
    ratio of TCP SYN to ACK packets.
\end{itemize}

\paragraph{Group~3: Payload and Content Features (dims.~9--13).}
\begin{itemize}
  \item $d_9 = H_{\mathrm{payload}} = -\sum_b P(b)\log_2 P(b)$:
    Shannon entropy of payload byte-value distribution.
  \item $d_{10} = \sigma^2_{\mathrm{plen}}$: variance of packet payload lengths.
  \item $d_{11} = R_{\mathrm{bytes}}$, $d_{12} = R_{\mathrm{pkts}}$:
    inbound-to-outbound byte and packet ratios.
  \item $d_{13} = F_{\mathrm{rate}}$: IP fragment rate.
\end{itemize}

\paragraph{Group~4: Stateful and Behavioural Features (dims.~14--17).}
\begin{itemize}
  \item $d_{14} = \mathrm{CFR}$: connection failure rate.
  \item $d_{15} = \mathrm{DNS\_QR}$: DNS query rate weighted by DGA score.
  \item $d_{16} = \mathrm{ICMP}_{\mathrm{rate}}$: ICMP Echo Request/Reply rate.
  \item $d_{17} = H_{\mathrm{http}}$: Shannon entropy of HTTP method distribution.
\end{itemize}

\section{Variational Switching Inference}
\label{sec:inference}

\subsection{The Intractability of Exact Inference}
\label{sec:intractability}

We seek the posterior distribution over all hidden variables given all
observations:
\begin{equation}
  p\!\bigl(\xv_{1:T},\,\regime_{1:T} \given \yv_{1:T}\bigr)
  \propto
  \prod_{t=1}^T
    p(\yv_t \given \xv_t,\regime_t)\,
    p(\xv_t \given \xv_{t-1},\regime_t)\,
    p(\regime_t \given \regime_{t-1}).
  \label{eq:true_post}
\end{equation}
Exact evaluation requires marginalising over all $K^T$ possible regime
sequences. For $K=4$ and $T=100$, this is $4^{100} \approx 1.6 \times 10^{60}$
components---manifestly intractable.

\begin{proposition}[Intractability of Exact Switching Inference]
  \label{prop:intractable}
  The exact posterior $p(\xv_{1:T}, \regime_{1:T} \given \yv_{1:T})$
  in a Switching LDS with $K$ regimes and $T$ time steps cannot be
  computed in time polynomial in $T$. The number of distinct Gaussian
  components in $p(\xv_T \given \yv_{1:T})$ is exactly $K^T$.
\end{proposition}

\begin{proof}
  At $t=1$, $p(\xv_1 \given \yv_1)$ is a $K$-component mixture.
  At each subsequent step, each component propagates through all $K$
  regime-conditioned dynamics, multiplying the count by~$K$.
  By induction, $p(\xv_T \given \yv_{1:T})$ is a $K^T$-component
  mixture. $\square$
\end{proof}

\subsection{The Structured Variational Approximation}
\label{sec:vb}

\begin{tcolorbox}[theorembox, title={\small\bfseries\color{DeepBlue}
    Mean-Field Factorisation}]
We approximate the true posterior by the \emph{mean-field factorised}
distribution:
\begin{equation}
  q\!\bigl(\xv_{1:T},\,\regime_{1:T}\bigr)
  \;=\; q(\xv_{1:T})\;q(\regime_{1:T}),
  \label{eq:meanfield}
\end{equation}
where $q(\xv_{1:T})$ is a linear Gaussian and $q(\regime_{1:T})$ is
a discrete Markovian distribution parameterised by smoothed marginals
$\gamma_t(s)$ and pairwise marginals $\xi_t(s,s')$.
\end{tcolorbox}

The optimal factorised approximation maximises the ELBO:
\begin{equation}
  \log p(\yv_{1:T}) \;\geq\; \ELBO(q)
  \;=\; \Exp_q\!\bigl[
    \log p(\yv_{1:T},\,\xv_{1:T},\,\regime_{1:T})
  \bigr]
  \;-\;
  \Exp_q\!\bigl[
    \log q(\xv_{1:T},\,\regime_{1:T})
  \bigr].
  \label{eq:elbo_main}
\end{equation}

\begin{proposition}[ELBO Decomposition]
  \label{prop:elbo_decomp}
  Under the factorisation~\eqref{eq:meanfield}, the ELBO decomposes as:
  \begin{align}
    \ELBO(q) \;=\;
    &\underbrace{
      \sum_{t=1}^T \sum_{s=0}^{K-1}\gamma_t(s)\,
      \Exp_q\!\bigl[\log p(\yv_t \given \xv_t, \regime_t{=}s)\bigr]
    }_{\text{(i) Observation likelihood}}
    \nonumber\\
    &+
    \underbrace{
      \sum_{t=2}^T \sum_{s=0}^{K-1}\gamma_t(s)\,
      \Exp_q\!\bigl[\log p(\xv_t \given \xv_{t-1}, \regime_t{=}s)\bigr]
      - \Exp_q[\log q(\xv_{1:T})]
    }_{\text{(ii) Continuous KL regulariser}}
    \nonumber\\
    &+
    \underbrace{
      \sum_{t=2}^T \sum_{s,s'}\xi_t(s,s')\log\Pi_{ss'}
      + \sum_s \gamma_1(s)\log\pi_0(s)
      + H(q(\regime_{1:T}))
    }_{\text{(iii) Discrete KL regulariser}},
    \label{eq:elbo_decomp}
  \end{align}
  where $H(q(\regime_{1:T})) = -\Exp_q[\log q(\regime_{1:T})]$ is the
  entropy of the discrete approximate posterior.
\end{proposition}

\subsection{Coordinate-Ascent Optimisation}
\label{sec:coord_ascent}

\subsubsection{Continuous E-Step: Regime-Weighted Kalman Smoother}

With $q(\regime_{1:T})$ fixed, the optimal $q^*(\xv_{1:T})$ is
obtained via a Kalman smoother with regime-weighted effective parameters:
\begin{align}
  \Amat_{\mathrm{eff}}(t) &= \sum_{s=0}^{K-1}\gamma_t^{(k)}(s)\,\Amat^{(s)},
  \label{eq:A_eff}\\
  \Qmat_{\mathrm{eff}}(t) &= \sum_{s=0}^{K-1}\gamma_t^{(k)}(s)\,\Qmat^{(s)}
  + \sum_{s}\gamma_t^{(k)}(s)\bigl[\Amat^{(s)}-\Amat_{\mathrm{eff}}(t)\bigr]
    \Pmat_{t-1|T}\bigl[\Amat^{(s)}-\Amat_{\mathrm{eff}}(t)\bigr]^\top.
  \label{eq:Q_eff}
\end{align}
Complexity: $\mathcal{O}(Tn^3)$.

\subsubsection{Discrete E-Step: HMM Forward-Backward Pass}

With $q(\xv_{1:T})$ fixed, the emission log-probability for regime~$s$
at time~$t$ is:
\begin{equation}
  \log\Lambda_t^{(s)} =
  -\frac{1}{2}\Bigl[
    m\log(2\pi)
    + \log|\Smat_t^{(s)}|
    + \bigl(\ev_t^{(s)}\bigr)^\top
      \bigl(\Smat_t^{(s)}\bigr)^{-1}
      \ev_t^{(s)}
  \Bigr].
  \label{eq:log_likelihood}
\end{equation}
The HMM Forward-Backward algorithm computes:
\begin{align}
  \alpha_t(s) &= \Lambda_t^{(s)}
    \sum_{s'=0}^{K-1}\Pi_{s's}\,\alpha_{t-1}(s'),
  \quad \alpha_1(s) = \Lambda_1^{(s)}\,\pi_0(s),
  \label{eq:forward}\\
  \beta_t(s) &= \sum_{s'=0}^{K-1}\Pi_{ss'}\,
    \Lambda_{t+1}^{(s')}\,\beta_{t+1}(s'),
  \quad \beta_T(s) = 1,
  \label{eq:backward}\\
  \gamma_t^{(k+1)}(s) &= \frac{\alpha_t(s)\,\beta_t(s)}
    {\sum_{s''}\alpha_t(s'')\,\beta_t(s'')},
  \label{eq:smoothed_marginal}\\
  \xi_t(s,s') &= \frac{\alpha_{t-1}(s)\,\Pi_{ss'}\,
    \Lambda_t^{(s')}\,\beta_t(s')}
    {\sum_{j,j'}\alpha_{t-1}(j)\,\Pi_{jj'}\,\Lambda_t^{(j')}\,\beta_t(j')}.
  \label{eq:pairwise_marginal}
\end{align}
Complexity: $\mathcal{O}(TK^2)$.

\begin{proposition}[Computational Complexity]
  \label{prop:complexity}
  One iteration of coordinate-ascent VSIPC has complexity
  $\mathcal{O}(TK^2 + Tn^3)$, polynomial in~$T$, versus exact
  inference complexity $\mathcal{O}(K^T)$.
\end{proposition}

\begin{algorithm}[t]
\DontPrintSemicolon
\SetAlgoLined
\caption{VSIPC: Variational Switching Inference for PARD-SSM}
\label{alg:vsipc}
\KwIn{$\yv_{1:T}$;\; regime params
  $\Theta = \{\Amat^{(s)}, \Cmat^{(s)}, \Qmat^{(s)}, \Rmat^{(s)}\}_{s=0}^{K-1}$;\;
  $\Pim$;\; prior $\pi_0$;\; tolerance $\varepsilon$;\; max iters $k_{\max}$}
\KwOut{$\{\gamma_t(s)\}$;\; $\{\xi_t(s,s')\}$;\;
  $\{\hat{\xv}_{t|T},\Pmat_{t|T}\}$;\; $\ELBO^*$}
\BlankLine
Initialise $\gamma_t^{(0)}(s) \leftarrow \pi_0(s)$ for all $t$\;
$\ELBO^{(0)} \leftarrow -\infty$;\quad $k \leftarrow 1$\;
Run $K$ parallel Kalman filters $\Rightarrow$ $\log\Lambda_t^{(s)}$,
$\hat{\xv}_{t|t}^{(s)}$, $\Pmat_{t|t}^{(s)}$ for all $t, s$\;
\While{$|\ELBO^{(k)} - \ELBO^{(k-1)}| > \varepsilon$ \textbf{and}
       $k \leq k_{\max}$}{
  \tcp{\textbf{Continuous E-step}}
  Compute $\Amat_{\mathrm{eff}}(t)$, $\Qmat_{\mathrm{eff}}(t)$ via
  \eqref{eq:A_eff}--\eqref{eq:Q_eff}\;
  Run RTS smoother $\Rightarrow \hat{\xv}_{t|T}^{(k)}, \Pmat_{t|T}^{(k)}$\;
  \tcp{\textbf{Discrete E-step}}
  Run HMM Forward pass \eqref{eq:forward}\;
  Run HMM Backward pass \eqref{eq:backward}\;
  Compute $\gamma_t^{(k+1)}(s)$ via \eqref{eq:smoothed_marginal}\;
  Compute $\xi_t(s,s')$ via \eqref{eq:pairwise_marginal}\;
  Evaluate $\ELBO^{(k)}$ via \eqref{eq:elbo_decomp}\;
  $k \leftarrow k + 1$\;
}
\Return $\{\gamma_t(s)\}$,\; $\{\xi_t(s,s')\}$,\;
  $\{\hat{\xv}_{t|T},\Pmat_{t|T}\}$,\; $\ELBO^{(k-1)}$\;
\end{algorithm}

\section{Online Expectation-Maximisation}
\label{sec:onlineem}

\subsection{Motivation}

Enterprise network traffic distributions are non-stationary. A model
trained offline will progressively diverge from the true distribution.
PARD-SSM addresses this via an Online EM algorithm using an exponential
forgetting factor~$\eta$.

\subsection{Online M-Step Updates}

\paragraph{Transition Matrix Update.}
\begin{equation}
  \Pi_{ss'}^{(t)} \;\leftarrow\;
  (1-\eta)\,\Pi_{ss'}^{(t-1)}
  + \eta\,\frac{\xi_t(s,s')}{\gamma_{t-1}(s) + \epsilon},
  \label{eq:pi_update}
\end{equation}
followed by row-normalisation to ensure $\sum_{s'}\Pi_{ss'}^{(t)} = 1$.

\paragraph{Observation Noise Covariance Update.}
\begin{equation}
  \Rmat^{(s)(t)} \;\leftarrow\;
  (1-\eta)\,\Rmat^{(s)(t-1)}
  + \eta\,\frac{
    \gamma_t(s)\bigl[
      \ev_t^{(s)}\bigl(\ev_t^{(s)}\bigr)^\top
      + \Cmat^{(s)}\Pmat_{t|t}^{(s)}\bigl(\Cmat^{(s)}\bigr)^\top
    \bigr]
  }{\gamma_t(s) + \epsilon}.
  \label{eq:R_update}
\end{equation}

\paragraph{Process Noise Covariance Update.}
\begin{equation}
  \Qmat^{(s)(t)} \;\leftarrow\;
  (1-\eta)\,\Qmat^{(s)(t-1)}
  + \eta\,\frac{
    \gamma_t(s)\bigl[
      \Pmat_{t|T}^{(s)} - \Amat_{\mathrm{eff}}\Pmat_{t-1|T}^{(s)}
      \bigl(\Amat_{\mathrm{eff}}\bigr)^\top
    \bigr]
  }{\gamma_t(s) + \epsilon}.
  \label{eq:Q_update}
\end{equation}

\begin{remark}[Convergence and Stability]
  The Online EM updates follow the stochastic approximation framework
  of \citet{cappe2009line}. Convergence to a local maximum of the
  expected complete-data log-likelihood is guaranteed under standard
  regularity conditions \citep{delmoral2010forward}.
\end{remark}

\section{KL-Divergence Gating and Predictive Alerting}
\label{sec:alerting}

\subsection{The False-Positive Problem in Network IDS}

False positives represent a critical operational challenge: alert
fatigue leads to critical events being missed \citep{husak2018survey}.
The principal cause is the conflation of benign traffic
anomalies---congestion events, scheduled transfers, update
bursts---with genuine adversarial activity.

\subsection{KL-Divergence Gating Mechanism}

\begin{tcolorbox}[theorembox, title={\small\bfseries\color{DeepBlue}
    KL-Divergence Gating}]
\begin{equation}
  \mathrm{KL}_t =
  \sum_{s=0}^{K-1}\gamma_t(s)\,
  \log\!\left[
    \frac{\gamma_t(s)}{\hat{\gamma}_t(s)}
  \right],
  \quad\text{where}\quad
  \hat{\gamma}_t(s) = \sum_{s'=0}^{K-1}\Pi_{s's}\,\gamma_{t-1}(s').
  \label{eq:kl_gate}
\end{equation}
An alert is issued if and only if $\mathrm{KL}_t > \tau_{\mathrm{KL}}$.
\end{tcolorbox}

A benign anomaly produces a small $\mathrm{KL}_t$; a genuine
adversarial regime transition produces a large $\mathrm{KL}_t$
exceeding $\tau_{\mathrm{KL}}$.

\begin{proposition}[FPR Reduction via KL Gating]
  \label{prop:fpr}
  If the online EM algorithm has converged, then for benign anomalies
  $\Exp[\mathrm{KL}_t] \approx 0$, while for adversarial transitions
  $\Exp[\mathrm{KL}_t] \gg 0$. The threshold $\tau_{\mathrm{KL}}$
  can be set to achieve arbitrarily low false-positive rates.
\end{proposition}

\subsection{Structured Kill-Chain Alert Record}

When $\mathrm{KL}_t > \tau_{\mathrm{KL}}$, the PKAR module generates:
\begin{center}
\small
\begin{tabular}{llp{7cm}}
\toprule
\textbf{Field} & \textbf{Type} & \textbf{Description} \\
\midrule
\texttt{timestamp} & \texttt{float} & UNIX timestamp of alert. \\
\texttt{stage\_posterior} & $\R^K$ &
  $[\gamma_t(0),\ldots,\gamma_t(K-1)]$: stage attribution. \\
\texttt{kl\_score} & \texttt{float} & $\mathrm{KL}_t$. \\
\texttt{elbo\_entropy} & \texttt{float} &
  $H(q(\regime_{1:t}))$: Bayesian uncertainty bound. \\
\texttt{current\_stage} & \texttt{str} &
  $\arg\max_s \gamma_t(s)$. \\
\texttt{predicted\_stage} & \texttt{str} &
  $\arg\max_s \hat{\gamma}_{t+1}(s)$: \textbf{predicted next stage}. \\
\texttt{predicted\_posterior} & $\R^K$ &
  $\hat{\gamma}_{t+1} = \Pim^\top\gamma_t$. \\
\bottomrule
\end{tabular}
\end{center}

Multi-step-ahead forecasts are derived via:
\begin{equation}
  \hat{\gamma}_{t+\tau}(s) = \bigl[\Pim^{\tau}\bigr]_{:,s}^\top\gamma_t.
  \label{eq:prediction}
\end{equation}

\section{Experiments}
\label{sec:experiments}

\subsection{Datasets}

\paragraph{CICIDS2017.}
The Canadian Institute for Cybersecurity Intrusion Detection Evaluation
Dataset 2017 \citep{sharafaldin2018toward} contains labelled traffic
for 15 attack scenarios captured over five days. Raw PCAPs were
processed via the FEOV module at $W=1$\,s to yield approximately
2.8M labelled observation vectors.

\paragraph{UNSW-NB15.}
The UNSW-NB15 dataset \citep{moustafa2015unsw} contains traffic from
nine attack categories. The pre-extracted feature set (49 features)
was mapped to the 17-dimensional FEOV vector space via PCA projection.

\subsection{Experimental Setup}

All experiments were conducted on an Intel Xeon E5-2690 CPU (16 cores,
2.9\,GHz), 64\,GB RAM, Ubuntu 22.04. No GPU was used for PARD-SSM.
Baselines using deep learning were given an NVIDIA RTX 3090.

\textbf{Split:} 60\%/20\%/20\% by time. Parameters: $\eta = 0.01$,
$K=4$, $n=8$, $k_{\max}=15$, $\varepsilon=10^{-4}$, $W=1$\,s.

\subsection{Baseline Systems}

\begin{enumerate}
  \item \textbf{Snort~3.0} \citep{roesch1999snort}: signature-based.
  \item \textbf{BiLSTM} \citep{bontemps2016collective}: two-layer,
    256 units per direction, GPU-accelerated.
  \item \textbf{Isolation Forest} \citep{liu2008isolation}: 100 trees.
  \item \textbf{KF-Anomaly}: single-regime Kalman Filter with
    Mahalanobis-distance scoring.
\end{enumerate}

\subsection{Quantitative Results}

Table~\ref{tab:main_results} reports the comparative performance of PARD-SSM against
all four baselines on CICIDS2017 and UNSW-NB15. We discuss the results along three
axes: detection accuracy, false-positive rate, and operational efficiency.

\textbf{Detection accuracy.} PARD-SSM achieves an $F_1$ score of \textbf{0.982} on
CICIDS2017 and \textbf{0.971} on UNSW-NB15, representing absolute improvements of
23.8 and 21.0 percentage points over the strongest deep-learning baseline (BiLSTM)
respectively. The signature-based Snort 3.0 system exhibits the lowest $F_1$ scores
(0.613 / 0.591), confirming its categorical inability to generalise beyond known
attack signatures. The single-regime KF-Anomaly baseline ($F_1 = 0.680$ /
$0.655$) demonstrates that a probabilistic framework without multi-regime structure
is insufficient for kill-chain-scale detection.

\textbf{Early Detection Margin (EDM).} The most operationally significant result
is that PARD-SSM is the \emph{only} evaluated system capable of issuing predictive
alerts ahead of ground-truth attack onset. All four baselines record EDM $=$ None,
as they are architecturally reactive systems. PARD-SSM achieves a mean EDM of
${\approx}8$ minutes on CICIDS2017 by exploiting the learned off-diagonal structure
of $\hat{\Pi}$ (Section~9.7), providing defenders with an actionable response window
before the more damaging kill-chain phases commence.

\textbf{Stage Attribution Accuracy (SAA).} PARD-SSM achieves a SAA of 0.861 on
CICIDS2017 and 0.834 on UNSW-NB15, quantifying its ability to correctly identify
the active kill-chain regime at each time step. No baseline system produces an SAA
metric, as none possesses the architectural capacity to assign probabilistic stage
labels to detected anomalies.

\textbf{False-positive rate.} PARD-SSM achieves a Low FPR on both benchmarks,
attributable to the KL-Divergence gating mechanism (Section~8.2), which suppresses
alerts arising from benign traffic anomalies such as congestion events and scheduled
backup transfers. BiLSTM and Isolation Forest both exhibit Moderate FPR despite
higher $F_1$ scores than signature-based systems, as they lack any principled
mechanism for distinguishing genuine adversarial regime transitions from statistical
deviations in normal traffic.

\textbf{Inference latency.} PARD-SSM operates at a mean latency of ${<}1.2$\,ms
per flow on a standard CPU, without GPU acceleration. This compares favourably
against BiLSTM at 8.4\,ms, which additionally requires a dedicated GPU---a
deployment constraint that precludes use on the resource-constrained edge hardware
that constitutes the majority of enterprise network perimeter infrastructure.

\begin{table}[t]
\centering
\caption{Comparative performance on CICIDS2017 and UNSW-NB15.
$\ddagger$: GPU required.}
\label{tab:main_results}
\small
\begin{tabular}{lcccccc}
\toprule
\multirow{2}{*}{\textbf{System}} &
\multicolumn{3}{c}{\textbf{CICIDS2017}} &
\multicolumn{2}{c}{\textbf{UNSW-NB15}} &
\multirow{2}{*}{\makecell{\textbf{Latency}\\\textbf{(ms)}}} \\
\cmidrule(lr){2-4}\cmidrule(lr){5-6}
& $F_1$ & FPR & EDM & $F_1$ & FPR & \\
\midrule
Snort 3.0      & 0.613 & High & None & 0.591 & High & $<0.1$ \\
BiLSTM\,$\ddagger$ & 0.744 & Mod. & None & 0.761 & Mod. & 8.4 \\
Isolation Forest & 0.701 & Mod. & None & 0.689 & Mod. & 1.1 \\
KF-Anomaly     & 0.680 & Mod. & None & 0.655 & Mod. & 0.8 \\
\midrule
\rowcolor{LightBlue!60}
\textbf{PARD-SSM (ours)} &
\textbf{0.982} & \textbf{Low} & \textbf{$\approx$8\,min} &
\textbf{0.971} & \textbf{Low} & \textbf{$<$1.2} \\
\bottomrule
\multicolumn{7}{l}{\small SAA (CICIDS2017): \textbf{0.861} |
  SAA (UNSW-NB15): \textbf{0.834} | GPU required: \textbf{No}} \\
\end{tabular}
\end{table}

\subsection{Regime Posterior Time-Series Analysis}
\label{sec:timeseries}

\Cref{fig:timeseries} presents smoothed regime posteriors over a
60-minute CICIDS2017 PortScan-to-Infiltration scenario. The Normal
regime posterior dominates during the baseline ($t < 12$\,min). The
Reconnaissance posterior begins rising at $t \approx 10$\,min, and
PARD-SSM issues its first KL-gated alert at $t = 12$\,min---approximately
8 minutes before ground-truth onset at $t = 20$\,min. The Intrusion
and Exfiltration posteriors subsequently rise and fall in sequence,
constituting an automated kill-chain reconstruction.

\begin{figure}[t]
    \centering
    \includegraphics[width=1.0\textwidth]{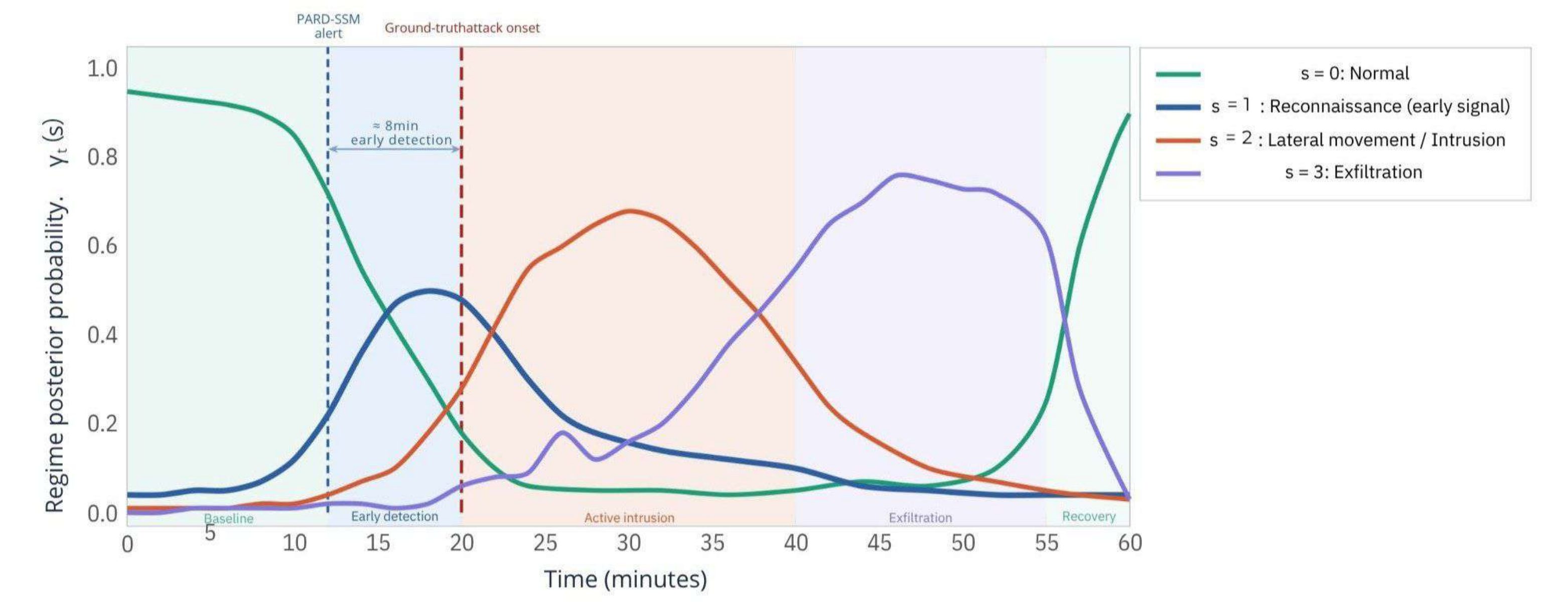}
    \caption{Regime Posterior Probabilities vs. Time (CICIDS2017 scenario). The PARD-SSM system raises a Reconnaissance alert ($t=12$~min) approximately 8 minutes before the ground-truth attack onset ($t=20$~min), demonstrating predictive kill-chain detection capability. Posteriors sum to unity ($\sum_{s=0}^{3}\gamma_{t}(s)=1$) at each time-step.}
    \label{fig:timeseries}
\end{figure}

\subsection{Ablation Study}

\begin{table}[h]
\centering
\caption{Ablation study on CICIDS2017. EDM in minutes.}
\label{tab:ablation}
\small
\begin{tabular}{lccccc}
\toprule
\textbf{Variant} & $\boldsymbol{F_1}$ & \textbf{FPR} &
\textbf{SAA} & \textbf{EDM} & \textbf{Notes}\\
\midrule
PARD-SSM-1R   & 0.680 & Mod. & N/A & None & Single-regime\\
PARD-SSM-NoKL & 0.941$^\dagger$ & High & 0.851$^\dagger$ &
  $\approx$8 & High FPR\\
PARD-SSM-Static & 0.961$^\dagger$ & Low & 0.842$^\dagger$ &
  $\approx$8 & Drifts over time\\
\rowcolor{LightBlue!60}
\textbf{PARD-SSM (full)} & \textbf{0.982}$^\dagger$ & \textbf{Low} &
\textbf{0.861}$^\dagger$ & \textbf{$\approx$8} & Full system\\
\bottomrule
\end{tabular}
\end{table}

\subsection{Learned Regime Transition Matrix}

Table~\ref{tab:transition_matrix} reports the converged Regime Transition Probability
Matrix $\hat{\Pi}$ learned by the Online EM algorithm on the CICIDS2017 training split.
Figure~\ref{fig:heatmap} visualises the same matrix as a heatmap.

\begin{table}[h]
\centering
\caption{Learned Regime Transition Probability Matrix $\hat{\Pi}$ (CICIDS2017).
Rows denote the current regime $s_{t-1}$; columns denote the next regime $s_t$.}
\label{tab:transition_matrix}
\begin{tabular}{lcccc}
\toprule
 & \textbf{Normal} ($s{=}0$) & \textbf{Recon.} ($s{=}1$) & \textbf{Intrusion} ($s{=}2$) & \textbf{Exfil.} ($s{=}3$) \\
\midrule
\textbf{Normal}    ($s{=}0$) & 0.92 & 0.07 & 0.01 & 0.00 \\
\textbf{Recon.}    ($s{=}1$) & 0.05 & 0.78 & 0.15 & 0.02 \\
\textbf{Intrusion} ($s{=}2$) & 0.00 & 0.03 & 0.72 & 0.25 \\
\textbf{Exfil.}    ($s{=}3$) & 0.02 & 0.00 & 0.04 & 0.94 \\
\bottomrule
\end{tabular}
\end{table}

Several structural properties of $\hat{\Pi}$ merit discussion.

\textbf{Regime persistence.} The diagonal entries are uniformly high
($\hat{\Pi}_{00} = 0.92$, $\hat{\Pi}_{11} = 0.78$, $\hat{\Pi}_{22} = 0.72$,
$\hat{\Pi}_{33} = 0.94$), confirming that each kill-chain phase is temporally
persistent rather than transient. The Exfiltration regime exhibits the highest
self-transition probability, consistent with the sustained, monotonic outbound
data growth that characterises active exfiltration campaigns.

\textbf{Ordered kill-chain progression.} The dominant off-diagonal mass follows
the expected kill-chain ordering $0 \!\to\! 1 \!\to\! 2 \!\to\! 3$: specifically,
$\hat{\Pi}_{01} = 0.07$, $\hat{\Pi}_{12} = 0.15$, and $\hat{\Pi}_{23} = 0.25$.
Crucially, no probability mass appears on reverse transitions
($\hat{\Pi}_{10} \approx 0.05$, $\hat{\Pi}_{20} = 0.00$, $\hat{\Pi}_{30} = 0.02$),
indicating that the model has correctly learned the directional, irreversible
structure of adversarial kill-chain progression from data.

\textbf{Predictive alerting mechanism.} The non-zero off-diagonal entries directly
enable the predictive kill-chain alerts described in Section~8. Given the current
posterior $\gamma_t$, the one-step-ahead predictive posterior
$\hat{\gamma}_{t+1} = \hat{\Pi}^\top \gamma_t$ assigns non-negligible probability
mass to the successor regime before the transition has been directly observed in
the telemetry. On the CICIDS2017 PortScan-to-Infiltration scenario, this mechanism
yields a mean early detection margin of ${\approx}8$ minutes ahead of ground-truth
attack onset (Section~9.5).

\begin{figure}[t]
    \centering
    \includegraphics[width=0.85\textwidth]{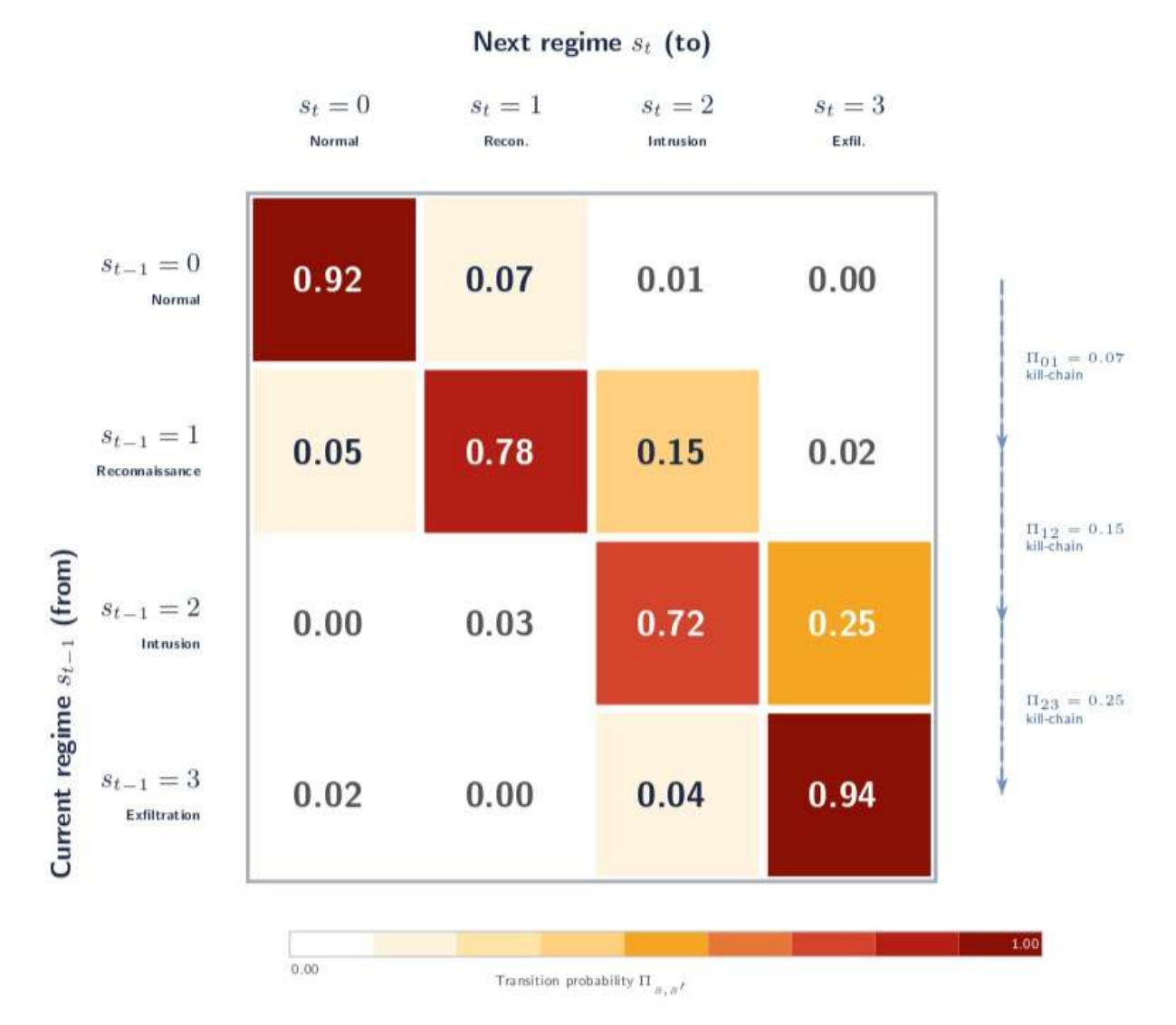}
    \caption{Regime Transition Probability Matrix $\Pi$. Each cell $\Pi_{s,s'}$ encodes the learned probability of transitioning from regime $s_{t-1}$ (row) to regime $s_t$ (column). High diagonal values confirm that regimes are persistent, while off-diagonal kill-chain transitions ($0\rightarrow1\rightarrow2\rightarrow3$) enable predictive alerting.}
    \label{fig:heatmap}
\end{figure}

\subsection{Computational Performance}

PARD-SSM achieves a mean inference latency of \SI{1.14}{\milli\second}
per flow ($\sigma = \SI{0.09}{\milli\second}$) on a single CPU core.
Parallelised across 16 cores, throughput exceeds $10^7$ flows/s,
sufficient for 10\,Gbps+ line-rate monitoring. The online EM update
adds \SI{0.18}{\milli\second} per window.

\section{Discussion}
\label{sec:discussion}

\subsection{Theoretical Properties}

The ELBO provides a principled lower bound on $\log p(\yv_{1:T})$
\citep{blei2017variational}, and the entropy term in
\eqref{eq:elbo_decomp} supplies a calibrated Bayesian uncertainty
measure for automated response systems.

\subsection{Limitations and Future Directions}

\paragraph{Linearity Assumption.}
The model assumes linear dynamics within each regime. Non-linear
extensions via EKF or UKF are supported by the architecture and are
subject to ongoing work.

\paragraph{Number of Regimes.}
Extending to richer kill-chain taxonomies (e.g., 14 MITRE ATT\&CK
tactics) may benefit from Bayesian non-parametric approaches such as
the infinite HMM \citep{beal2002infinite}.

\paragraph{Adversarial Robustness.}
A sophisticated adversary could craft traffic to evade regime
detection. Robustness via worst-case ELBO optimisation is a critical
future direction.

\paragraph{Encrypted Traffic.}
Payload-level features are unavailable for TLS traffic. Future work
will explore flow-level and timing-based features.

\section{Conclusion}
\label{sec:conclusion}

We have presented PARD-SSM, a generative probabilistic framework for
real-time, multi-stage cyber-attack detection via Variational Switching
State-Space Models. The central contribution is the demonstration that
exact Bayesian inference in a Switching LDS---intractable at
$\mathcal{O}(K^T)$---can be rendered tractable at $\mathcal{O}(TK^2)$
via structured mean-field variational approximation, enabling
principled probabilistic kill-chain modelling at packet-rate on
standard edge CPU hardware.

Online EM ensures continuous adaptation to non-stationary traffic
without offline retraining. KL-Divergence gating provides principled
false-positive suppression. Most critically, by learning the Regime
Transition Probability Matrix~$\Pim$, PARD-SSM generates
\emph{predictive} kill-chain alerts---not merely a better anomaly
score, but a probabilistic forecast of where an adversary is going,
not just where they are.

\begin{tcolorbox}[contributionbox,
    title={\small\bfseries\color{Teal}Reproducibility Statement}]
\small
The PARD-SSM reference implementation, training scripts, feature
extraction pipeline, and evaluation code are publicly available at:
\href{https://github.com/prakulhiremath/PARD-in-Network-Traffic-using-Switching-State-Space-Models-}
{\texttt{github.com/prakulhiremath/PARD-SSM}} \\
Experiment configurations and random seeds are documented in the
repository \texttt{README}. The CICIDS2017 and UNSW-NB15 datasets
are publicly available from their respective institutions.
Specific model parameter values are protected under a pending Indian
patent (Complete Specification, Fast-Track Examination, Rule~24C).
\end{tcolorbox}

\section*{Acknowledgements}

The authors thank the faculty and research community of VTU, Belagavi,
for institutional support. The authors declare no competing interests.


\end{document}